\documentclass[10pt,a4paper]{article}

\usepackage{latexsym,amsfonts,amsmath,amssymb,mathrsfs,url,color,amsthm}
\usepackage{graphicx,cite}

\newtheorem{theorem}{Theorem}

\newtheorem{remark}{Remark}

\usepackage[dvipdfm,
linkbordercolor={1 0 0},
citebordercolor={0 1 0},
urlbordercolor={0 1 1}]{hyperref}

\newcommand{\rd}{\,\mathrm{d}}
\newcommand{\EE}{\mathbb{E}}

\newcommand{\RR}{\mathbb{R}}
\newcommand{\Var}{\mathrm{Var}}
\allowdisplaybreaks

\begin{document}

\title{Computing the variance of a conditional expectation via non-nested Monte Carlo\thanks{The work of T.~G. is supported by JSPS Grant-in-Aid for Young Scientists No.~15K20964.}}

\author{Takashi Goda\thanks{Graduate School of Engineering, The University of Tokyo, 7-3-1 Hongo, Bunkyo-ku, Tokyo 113-8656, Japan. ({\tt goda@frcer.t.u-tokyo.ac.jp})}}

\date{\today}

\maketitle

\begin{abstract}
Computing the variance of a conditional expectation has often been of importance in uncertainty quantification.
Sun et al.\ has introduced an unbiased nested Monte Carlo estimator, which they call $1\frac{1}{2}$-level simulation since the optimal inner-level sample size is bounded as the computational budget increases.
In this letter we construct unbiased non-nested Monte Carlo estimators based on the so-called pick-freeze scheme due to Sobol'.
An extension of our approach to compute higher order moments of a conditional expectation is also discussed.
\end{abstract}
\emph{Keywords}: variance of a conditional expectation, Monte Carlo estimation, pick-freeze scheme, bias correction

%%%%%%%%%%%%%%%%%%%%%%%%%%%%%%%%%%%%%%%%%%%%%%%%%%%%%%%%%%%%%%%%%%%%%%%%%%%%%%%%%%%%%%%%%%%%%%%%%%%
%%%%%%%%%%%%%%%%%%%%%%%%%%%%%%%%%%%%%%%%%%%%%%%%%%%%%%%%%%%%%%%%%%%%%%%%%%%%%%%%%%%%%%%%%%%%%%%%%%%
%%%%%%%%%%%%%%%%%%%%%%%%%%%%%%%%%%%%%%%%%%%%%%%%%%%%%%%%%%%%%%%%%%%%%%%%%%%%%%%%%%%%%%%%%%%%%%%%%%%
\section{Introduction}\label{sec:1}
Let $X$ be a random variable with probability density function $p_X$ defined on $\Omega_X$, and let $f\colon \Omega_X\to \RR$ be a function.
For another random variable $Y$ which is correlated with $X$, we are interested in computing the \emph{variance of a conditional expectation}
\begin{align}
\Var_Y\left[\EE_{X|Y}[f]\right] & := \int_{\Omega_Y}\left( \int_{\Omega_X}f(x)p_{X|Y=y}(x)\rd x -\mu \right)^2 p_Y(y) \rd y \label{eq:Var_E1} \\
& \: = \int_{\Omega_Y}\left( \int_{\Omega_X}f(x)p_{X|Y=y}(x)\rd x\right)^2 p_Y(y) \rd y -\mu^2 ,\label{eq:Var_E2}
\end{align}
where $p_Y$ and $p_{X|Y=y}$ denote the probability density function (defined on $\Omega_Y$) of $Y$ and the conditional probability density function of $X$ given $Y=y$, respectively, and further $\mu$ is defined by
\begin{align*}
\mu := \int_{\Omega_Y}\int_{\Omega_X}f(x)p_{X|Y=y}(x)p_Y(y) \rd x \rd y = \int_{\Omega_X}f(x)p_X(x)\rd x.
\end{align*}
It follows from the well-known variance decomposition formula that 
\begin{align}\label{eq:decomp}
\Var_Y\left[\EE_{X|Y}[f]\right] = \Var_X[f] - \EE_{Y}\left[ \Var_{X|Y}[f]\right],
\end{align}
where each term on the right-hand side is defined similarly.

The quantity $\Var_Y\left[\EE_{X|Y}[f]\right]$ has been used in the area of uncertainty quantification.
For instance, in \cite{ZW03}, Zouaoui and Wilson used $\Var_Y\left[\EE_{X|Y}[f]\right]$ as a quality measure of uncertainty on the mean time in a single-server queueing system due to uncertainty on the parameters.
Another usage of $\Var_Y\left[\EE_{X|Y}[f]\right]$ which we have in mind is as follows:
Assume that $X$ denotes a set of uncertain simulation inputs and $Y$ does a sample observation data.
In the absence of the data $Y$, the prior variance $\Var_X[f]$ measures the uncertainty of a simulation output $f$.
If the data $Y$ is available, the uncertainty of $f$ after knowing $Y=y$ is represented by the posterior variance $\Var_{X|Y=y}[f]$.
Thus the uncertainty of $f$ is expected to be reduced to $\EE_{Y}\left[ \Var_{X|Y}[f]\right]$ by obtaining the data $Y$.
It can be seen from the identity (\ref{eq:decomp}) that $\Var_Y\left[\EE_{X|Y}[f]\right]$ quantifies how much the uncertainty of $f$ can be reduced before and after obtaining the data $Y$.

By the definition (\ref{eq:Var_E1}), it seems natural to use a nested Monte Carlo estimator for $\Var_Y\left[\EE_{X|Y}[f]\right]$.
Recently in \cite{SAS11}, Sun et al. have introduced the following unbiased nested Monte Carlo estimator: 
For positive integers $K$ and $n_1,\ldots,n_K$, let $y_1,\ldots,y_K$ be sampled independently from $p_Y$, and for $k=1,\ldots,K$, let $x_{1k},\ldots,x_{n_k k}$ be sampled independently but conditionally from $p_{X|Y=y_k}$.
Moreover let $C=n_1+\cdots + n_K$,
\begin{align*}
\overline{f}_k = \frac{1}{n_k}\sum_{j=1}^{n_k}f(x_{jk})\quad  \text{and}\quad \overline{\overline{f}} = \frac{1}{C}\sum_{k=1}^{K}\sum_{j=1}^{n_k}f(x_{jk}).
\end{align*}
Then the quantity of interest $\Var_Y\left[\EE_{X|Y}[f]\right]$ is estimated by
\begin{align}\label{eq:sun_et_al}
W := \frac{SS_{\tau}-\frac{K-1}{C-K}SS_{\epsilon}}{C-\sum_{k=1}^{K}n_k^2/C},
\end{align}
where
\begin{align*}
SS_{\tau}= \sum_{k=1}^{K}n_k (\overline{f}_k-\overline{\overline{f}})^2\quad \text{and}\quad SS_{\epsilon}= \sum_{k=1}^{K}\sum_{j=1}^{n_k}(f(x_{jk})-\overline{f}_k)^2.
\end{align*}
The interesting property of the estimator $W$ is that the optimal inner-level sample sizes $n_1,\ldots,n_K$ remain bounded above as the total computational budget $C$ increases.
That is, there is no need to increase both the inner- and outer-level sample sizes simultaneously for making the approximation error converge to zero, and thus, it can be inferred from \cite[Equation~(10)]{SAS11} that the estimator $W$ achieves the Monte Carlo root mean square error (rmse) of order $C^{-1/2}$ asymptotically.
Because of this nice property, Sun et al.\ have referred to their estimator as $1\frac{1}{2}$-level simulation.
However, it may require a precomputation step to choose proper inner-level sample sizes depending on a problem at hand and a given cost $C$.

In this letter, by assuming that i.i.d.\ samplings from $p_Y$ and $p_{X|Y=y}$ for any $y\in \Omega_Y$ are possible, we construct several unbiased \emph{non-nested} Monte Carlo estimators for $\Var_Y\left[\EE_{X|Y}[f]\right]$.
We also show that our approach can be extended in a straightforward way to compute higher order moments of a conditional expectation.
We note that our assumption is same as that considered in \cite{SAS11}.
Since our estimators are no longer of the nested form, we do not need to take care of a proper choice of inner-level sample sizes, and our estimators are naturally expected to achieve the Monte Carlo rmse of order $C^{-1/2}$.
Our idea for constructing non-nested estimators stems from the pick-freeze scheme due to Sobol' \cite{S90,S01}, which was originally introduced for computing variance-based sensitivity indices and has been thoroughly studied in the context of global sensitivity analysis by Saltelli \cite{Sal02}, Owen \cite{O13}, Janon et al.\ \cite{JKLNP14}, and Owen et al.\ \cite{ODC14} to list just a few.
In fact, in that context, the quantity $\Var_Y\left[\EE_{X|Y}[f]\right]$ corresponds to the so-called first order sensitivity index, if $Y$ denotes a subset of uncertain simulation inputs contained in $X$.
Thus our result of this letter can be regarded as a generalization of the known results on variance-based sensitivity analysis.

The remainder of this letter is organized as follows.
In the next section, we introduce four straightforward non-nested Monte Carlo estimators; one based on Mauntz \cite{M02} and Kucherenko et al.\ \cite{KFSM11} is unbiased whereas the other three essentially based on Janon et al.\ \cite{JKLNP14} is biased.
In the third section, we give bias corrections of the latter estimators.
In the fourth section, we discuss an extension of our approach to compute higher order moments.
We conclude this letter with numerical experiments in the last section.
%%%%%%%%%%%%%%%%%%%%%%%%%%%%%%%%%%%%%%%%%%%%%%%%%%%%%%%%%%%%%%%%%%%%%%%%%%%%%%%%%%%%%%%%%%%%%%%%%%%
%%%%%%%%%%%%%%%%%%%%%%%%%%%%%%%%%%%%%%%%%%%%%%%%%%%%%%%%%%%%%%%%%%%%%%%%%%%%%%%%%%%%%%%%%%%%%%%%%%%
%%%%%%%%%%%%%%%%%%%%%%%%%%%%%%%%%%%%%%%%%%%%%%%%%%%%%%%%%%%%%%%%%%%%%%%%%%%%%%%%%%%%%%%%%%%%%%%%%%%
\section{Non-nested Monte Carlo}\label{sec:2}
The key ingredient of the pick-freeze scheme lies in how to deal with the square appearing in the first and second terms of (\ref{eq:Var_E2}).
It is easy to see from Fubini's theorem that we can rewrite the first term of (\ref{eq:Var_E2}) into
\begin{align}
& \int_{\Omega_Y}\left( \int_{\Omega_X}f(x)p_{X|Y=y}(x)\rd x\right) \left( \int_{\Omega_X}f(x')p_{X|Y=y}(x')\rd x'\right) p_Y(y) \rd y \nonumber \\
= & \int_{\Omega_Y}\int_{\Omega_X} \int_{\Omega_X} f(x)f(x')p_{X|Y=y}(x)p_{X|Y=y}(x') p_Y(y) \rd x \rd x' \rd y. \label{eq:Var_E2_1}
\end{align}
The second term of (\ref{eq:Var_E2}), i.e., the squared expectation $\mu^2$, can be simply rewritten into
\begin{align*}
\mu^2 & = \left(\int_{\Omega_Y}\int_{\Omega_X}f(x)p_{X|Y=y}(x)p_Y(y) \rd x \rd y\right)\left(\int_{\Omega_X}f(x'')p_X(x'')\rd x''\right) \\
& = \int_{\Omega_Y}\int_{\Omega_X}\int_{\Omega_X}f(x)f(x'')p_{X|Y=y}(x)p_X(x'')p_Y(y) \rd x \rd x''\rd y ,
\end{align*}
where we used Fubini's theorem in the second equality.
Thus the quantity $\Var_Y\left[\EE_{X|Y}[f]\right]$ is given by
\begin{align*}
\Var_Y\left[\EE_{X|Y}[f]\right] & = \int_{\Omega_Y}\int_{\Omega_X} \int_{\Omega_X} \int_{\Omega_X} f(x)\left( f(x')-f(x'')\right) \\
& \quad \times p_{X|Y=y}(x)p_{X|Y=y}(x') p_X(x'') p_Y(y) \rd x \rd x' \rd x'' \rd y.
\end{align*}
Therefore, our first Monte Carlo estimator, which has some similarity to that in \cite{M02,KFSM11} for variance-based sensitivity analysis, can be constructed as
\begin{align*}
U:=\frac{1}{N}\sum_{n=1}^{N}f(x_n)\left( f(x'_n)-f(x''_n)\right),
\end{align*}
where, for each $n$, we first sample $y_n$ randomly from $p_Y$ and then sample $x_n$ and $x'_n$ independently and randomly from $p_{X|Y=y_n}$.
Further we sample $x''_n$ randomly from $p_X$, or we first sample $y''_n$ randomly from $p_Y$ (independently of $y_n$) and then sample $x''_n$ randomly from $p_{X|Y=y''_n}$.
It is obvious that $\EE[U]=\Var_Y\left[\EE_{X|Y}[f]\right]$, meaning that the estimator $U$ is unbiased.

Since the estimator $U$ requires three function evaluations for each $n$, the total computational budget $C$ equals $3N$. It is further possible to construct Monte Carlo estimators which require two function evaluations for each $n$, i.e., $C=2N$.
Let us consider an approximation of $\mu$ instead of $\mu^2$.
This can be done by using the samples $x_n$'s and $x'_n$'s commonly as either
\begin{align*}
\hat{\mu} = \frac{1}{N}\sum_{n=1}^{N}f(x_n) \quad \text{or} \quad \hat{\mu}' = \frac{1}{N}\sum_{n=1}^{N}f(x'_n) \quad \text{or} \quad \frac{\hat{\mu}+\hat{\mu}'}{2}.
\end{align*}
Using these estimators for $\mu$, we can introduce the following Monte Carlo estimators for $\Var_Y\left[\EE_{X|Y}[f]\right]$:
\begin{align*}
V_1 & := \frac{1}{N}\sum_{n=1}^{N}f(x_n)f(x'_n) - \hat{\mu}^2, \\
V_2 & := \frac{1}{N}\sum_{n=1}^{N}f(x_n)f(x'_n) - \left(\frac{\hat{\mu}+\hat{\mu}'}{2}\right)^2 = \frac{1}{N}\sum_{n=1}^{N}\left( f(x_n)- \frac{\hat{\mu}+\hat{\mu}'}{2}\right)\left( f(x'_n) - \frac{\hat{\mu}+\hat{\mu}'}{2}\right), \\
V_3 & := \frac{1}{N}\sum_{n=1}^{N}f(x_n)f(x'_n) - \hat{\mu}\hat{\mu}' = \frac{1}{N}\sum_{n=1}^{N}\left( f(x_n)- \hat{\mu}\right)\left( f(x'_n) - \hat{\mu}'\right).
\end{align*}
Note that the last two estimators are exactly of the same forms as those in \cite{JKLNP14} for variance-based sensitivity analysis.
These estimators are actually the special cases of a generalized estimator 
\begin{align}\label{eq:V_est}
V := \frac{1}{N}\sum_{n=1}^{N}f(x_n)f(x'_n) - \left(w_1\hat{\mu}^2+w_2\hat{\mu}'^2+w_3\hat{\mu}\hat{\mu}'\right),
\end{align}
with real-valued parameters $w_1,w_2,w_3$ such that $w_1+w_2+w_3=1$.
Unfortunately, the following theorem states that all of the estimators $V_1,V_2,V_3$ are biased.

\begin{theorem}\label{thm:bias}
For reals $w_1,w_2,w_3$ such that $w_1+w_2+w_3=1$, we have
\begin{align*}
\EE\left[V\right] & = \Var_Y\left[\EE_{X|Y}[f]\right] - \frac{(w_1+w_2)\Var_X[f]+w_3\Var_Y\left[\EE_{X|Y}[f]\right]}{N} \\
& = \Var_Y\left[\EE_{X|Y}[f]\right] - \frac{\Var_X[f]-w_3\EE_Y\left[\Var_{X|Y}[f]\right]}{N}.
\end{align*}
In particular, we have
\begin{align*}
\EE\left[V_1\right] & = \Var_Y\left[\EE_{X|Y}[f]\right] - \frac{\Var_X[f]}{N}, \\
\EE\left[V_2\right] & = \Var_Y\left[\EE_{X|Y}[f]\right] - \frac{\Var_X[f]+\Var_Y\left[\EE_{X|Y}[f]\right]}{2N}, \\
\EE\left[V_3\right] & = \Var_Y\left[\EE_{X|Y}[f]\right] - \frac{\Var_Y\left[\EE_{X|Y}[f]\right]}{N}.
\end{align*}
\end{theorem}

\begin{proof}
First we have
\begin{align*}
& \quad \EE\left[\frac{1}{N}\sum_{n=1}^{N}f(x_n)f(x'_n)\right] = \frac{1}{N}\sum_{n=1}^{N}\EE\left[f(x_n)f(x'_n)\right] \\
& = \frac{1}{N}\sum_{n=1}^{N}\int_{\Omega_Y}\int_{\Omega_X} \int_{\Omega_X} f(x)f(x')p_{X|Y=y}(x)p_{X|Y=y}(x') p_Y(y) \rd x \rd x' \rd y\\
& = \int_{\Omega_Y}\int_{\Omega_X} \int_{\Omega_X} f(x)f(x')p_{X|Y=y}(x)p_{X|Y=y}(x') p_Y(y) \rd x \rd x' \rd y\\
& = \Var_Y\left[\EE_{X|Y}[f]\right]+\mu^2,
\end{align*}
where the second equality stems from the fact that $x_n$ and $x'_n$ are sampled independently and randomly from the same probability density $p_{X|Y=y_n}$ for a randomly sampled $y_n$.
Since the quantity shown in the third line is nothing but the first term of (\ref{eq:Var_E2}), see (\ref{eq:Var_E2_1}), we have the last equality.
In what follows, we focus on the expected value of the second term of $V$.

Since $x_1,\ldots,x_N$ are i.i.d.\ random samples from $p_X$, we have
\begin{align*}
& \quad \EE\left[\hat{\mu}^2\right] = \frac{1}{N^2}\sum_{m,n=1}^{N}\EE\left[f(x_m)f(x_n)\right] \\
& = \frac{1}{N^2}\sum_{n=1}^{N}\EE\left[\left(f(x_n)\right)^2\right] +  \frac{1}{N^2}\sum_{\substack{m,n=1 \\ m\neq n}}^{N}\EE\left[f(x_m)\right]\EE\left[f(x_n)\right] \\
& = \frac{1}{N^2}\sum_{n=1}^{N}\left( \Var_X[f]+\mu^2 \right) + \frac{1}{N^2}\sum_{\substack{m,n=1 \\ m\neq n}}^{N}\mu^2 = \mu^2+\frac{\Var_X[f]}{N}.
\end{align*}
The same argument gives
\begin{align*}
\EE\left[\hat{\mu}'^2\right] = \mu^2+\frac{\Var_X[f]}{N}.
\end{align*}
Moreover, in a similar way, we obtain
\begin{align*}
& \quad \EE\left[ \hat{\mu}\hat{\mu}' \right] = \frac{1}{N^2}\sum_{m,n=1}^{N}\EE\left[f(x_m)f(x'_n)\right] \\
& = \frac{1}{N^2}\sum_{n=1}^{N}\EE\left[f(x_n)f(x'_n)\right] +  \frac{1}{N^2}\sum_{\substack{m,n=1 \\ m\neq n}}^{N}\EE\left[f(x_m)\right]\EE\left[f(x'_n)\right] \\
& = \frac{1}{N^2}\sum_{n=1}^{N}\left(\Var_Y\left[\EE_{X|Y}[f]\right] +\mu^2 \right) + \frac{1}{N^2}\sum_{\substack{m,n=1 \\ m\neq n}}^{N}\mu^2 = \mu^2+\frac{\Var_Y\left[\EE_{X|Y}[f]\right]}{N}.
\end{align*}
Thus in total we have
\begin{align*}
\EE\left[w_1\hat{\mu}^2+w_2\hat{\mu}'^2+w_3\hat{\mu}\hat{\mu}'\right] = \mu^2+\frac{(w_1+w_2)\Var_X[f]+w_3\Var_Y\left[\EE_{X|Y}[f]\right]}{N},
\end{align*}
from which the first equality for $V$ in the theorem obviously follows.
The second equality for $V$ in the theorem follows from the equality $w_1+w_2+w_3=1$ and the identity (\ref{eq:decomp}).
The results for $V_1,V_2,V_3$ can be easily proven by choosing $w_1,w_2,w_3$ accordingly.
\end{proof}

Since it follows from the equation (\ref{eq:decomp}) that $0<\Var_Y\left[\EE_{X|Y}[f]\right] \leq \Var_X[f]$, the bias of $V_3$ is smaller or equal to that of $V_2$ in absolute value, which itself is smaller or equal to that of $V_1$.
Thus we recommend to use the estimator $V_3$ especially for small $N$.
If we restrict ourselves to the case $w_1,w_2,w_3\geq 0$, the estimator $V_3$ is optimal in this regard among possible realizations of $V$.
If such a restriction is not taken into account and $w_3$ can be set closed to $\Var_X[f]/\EE_Y\left[\Var_{X|Y}[f]\right] (>1)$, the bias can be made smaller.
Note, however, that the bias for every realized estimator decays at the rate $N^{-1}$, which is faster than that of the Monte Carlo rmse $N^{-1/2}$, so that the bias is negligible for large $N$.

%%%%%%%%%%%%%%%%%%%%%%%%%%%%%%%%%%%%%%%%%%%%%%%%%%%%%%%%%%%%%%%%%%%%%%%%%%%%%%%%%%%%%%%%%%%%%%%%%%%
%%%%%%%%%%%%%%%%%%%%%%%%%%%%%%%%%%%%%%%%%%%%%%%%%%%%%%%%%%%%%%%%%%%%%%%%%%%%%%%%%%%%%%%%%%%%%%%%%%%
%%%%%%%%%%%%%%%%%%%%%%%%%%%%%%%%%%%%%%%%%%%%%%%%%%%%%%%%%%%%%%%%%%%%%%%%%%%%%%%%%%%%%%%%%%%%%%%%%%%
\section{Bias correction}\label{sec:3}
Here we give a bias correction of the estimator $V$ for $N\geq 2$, from which bias corrections of the estimators $V_1,V_2,V_3$ are given directly.
Let us denote
\begin{align*}
s^2 = \frac{1}{N-1}\sum_{n=1}^{N}\left( f(x_n)-\hat{\mu}\right)^2 \quad \text{and}\quad s'^2 = \frac{1}{N-1}\sum_{n=1}^{N}\left( f(x'_n)-\hat{\mu}'\right)^2.
\end{align*}
When $w_3\neq N$, a bias corrected estimator of $V$ is given by
\begin{align*}
\tilde{V} := \frac{N}{N-w_3}\left(V+\frac{w_1s^2+w_2s'^2}{N}\right).
\end{align*}
\begin{theorem}\label{thm:unbias}
\begin{align*}
\EE\left[\tilde{V}\right] = \Var_Y\left[\EE_{X|Y}[f]\right].
\end{align*}
\end{theorem}

\begin{proof}
Noting again that $x_1,\ldots,x_N$ are i.i.d.\ random samples from $p_X$, it is well known that $\EE[s^2]=\EE[s'^2]=\Var_X[f]$.
Thus by using the linearity of the expectation and the result in Theorem~\ref{thm:bias} we have
\begin{align*}
& \quad \EE\left[\tilde{V}\right] = \frac{N}{N-w_3}\left(\EE\left[V\right]+\frac{w_1\EE\left[s^2\right]+w_2\EE\left[s'^2\right]}{N}\right) \\
& = \frac{N}{N-w_3}\left(\Var_Y\left[\EE_{X|Y}[f]\right] - \frac{(w_1+w_2)\Var_X[f]+w_3\Var_Y\left[\EE_{X|Y}[f]\right]}{N}+\frac{w_1+w_2}{N}\Var_X[f]\right) \\
& = \Var_Y\left[\EE_{X|Y}[f]\right].
\end{align*}
Hence the result follows.
\end{proof}

Choosing $w_1,w_2,w_3$ accordingly, bias corrections of the estimators $V_1,V_2,V_3$ are given as follows:
\begin{align*}
\tilde{V}_1 & := V_1+\frac{s^2}{N}, \\
\tilde{V}_2 & := \frac{2N}{2N-1}\left( V_2+\frac{s^2+s'^2}{4N}\right), \\
\tilde{V}_3 & := \frac{N}{N-1}V_3.
\end{align*}
The bias correction of $V_2$ is of the same form due to Owen \cite{O13}, which was given in the context of variance-based sensitivity analysis.
We would note that the corrected estimators $\tilde{V}_1,\tilde{V}_2,\tilde{V}_3$ still hold the important property of $V_1,V_2,V_3$, i.e., they require two function evaluations for each $n$ (instead of three as in $U$).

\begin{remark}
The estimator $\tilde{V}$ with the choice
\begin{align*}
w_1=w_2=\frac{N-1}{2(2N-1)} \quad \text{and}\quad w_3=\frac{N}{2N-1}
\end{align*}
is nothing but the estimator $W$ given in (\ref{eq:sun_et_al}) with $n_k=2$ for all $k$.
Since $w_1=w_2\to 1/4$ and $w_3\to 1/2$ when $N\to \infty$, the estimator $W$ with $n_k=2$ asymptotically coincides with the estimator $\tilde{V}_2$, which itself asymptotically coincides with the estimator $V_2$.
\end{remark}
%%%%%%%%%%%%%%%%%%%%%%%%%%%%%%%%%%%%%%%%%%%%%%%%%%%%%%%%%%%%%%%%%%%%%%%%%%%%%%%%%%%%%%%%%%%%%%%%%%%
%%%%%%%%%%%%%%%%%%%%%%%%%%%%%%%%%%%%%%%%%%%%%%%%%%%%%%%%%%%%%%%%%%%%%%%%%%%%%%%%%%%%%%%%%%%%%%%%%%%
%%%%%%%%%%%%%%%%%%%%%%%%%%%%%%%%%%%%%%%%%%%%%%%%%%%%%%%%%%%%%%%%%%%%%%%%%%%%%%%%%%%%%%%%%%%%%%%%%%%
\section{Extension to compute higher order moments}\label{sec:4}
Here we show that our approach can be extended to construct non-nested Monte Carlo estimators for higher order moments of a conditional expectation.
For an integer $m\geq 3$, the $m$-th crude moment and central moment are given by
\begin{align*}
\EE_Y\left[\left(\EE_{X|Y}[f]\right)^m\right] := \int_{\Omega_Y}\left( \int_{\Omega_X}f(x)p_{X|Y=y}(x)\rd x \right)^m p_Y(y) \rd y,
\end{align*}
and
\begin{align*}
\EE_Y\left[\left(\EE_{X|Y}[f]-\mu \right)^m\right] := \int_{\Omega_Y}\left( \int_{\Omega_X}f(x)p_{X|Y=y}(x)\rd x -\mu \right)^m p_Y(y) \rd y,
\end{align*}
respectively.
In what follows, we only consider an extension of the estimator $V_3$ and its bias correction for sake of simplicity.

Let us consider a random variable
\begin{align*}
W_m := \frac{1}{N}\sum_{n=1}^{N}\prod_{j=1}^{m}f(x_n^{(j)}) ,
\end{align*}
where, for each $n$, we first sample $y_n$ randomly from $p_Y$ and then sample $x_n^{(j)}$ independently and randomly from $p_{X|Y=y_n}$ for $j=1,\ldots,m$.
Then $W_m$ can be easily shown to be an unbiased estimator of the $m$-th crude moment as follows:
\begin{align*}
\EE\left[W_m\right] & = \frac{1}{N}\sum_{n=1}^{N}\EE\left[\prod_{j=1}^{m}f(x_n^{(j)})\right] \\
& = \frac{1}{N}\sum_{n=1}^{N}\int_{\Omega_Y}\prod_{j=1}^{m}\left( \int_{\Omega_X}f(x^{(j)})p_{X|Y=y}(x^{(j)})\rd x^{(j)} \right) p_Y(y) \rd y \\
& = \int_{\Omega_Y}\left( \int_{\Omega_X}f(x)p_{X|Y=y}(x)\rd x \right)^m p_Y(y) \rd y = \EE_Y\left[\left(\EE_{X|Y}[f]\right)^m\right].
\end{align*}
Furthermore, for $j=1,\ldots,m$, we denote the sample mean by
\begin{align*}
\hat{\mu}^{(j)} = \frac{1}{N}\sum_{n=1}^{N}f(x_n^{(j)}).
\end{align*}
As a generalization of the estimator $V_3$, a non-nested Monte Carlo estimator for the $m$-th central moment is given by
\begin{align*}
Z_m  :=  \frac{1}{N}\sum_{n=1}^{N}\prod_{j=1}^{m}\left( f(x_n^{(j)})- \hat{\mu}^{(j)}\right).
\end{align*}
Note that both the estimators $W_m$ and $Z_m$ require $m$ function evaluations for each $n$.
Although $Z_m$ is a biased estimator for the $m$-th central moment, it can be corrected as follows.
For the case $m=3$, the corrected estimator is given by
\begin{align*}
\tilde{Z}_3  :=  \frac{N^2}{(N-1)(N-2)}Z_3 ,
\end{align*}
for $N>2$. Furthermore, for the case $m=4$, the corrected estimator is given by
\begin{align*}
\tilde{Z}_4  :=  \frac{N^2}{(N-1)(N-2)(N-3)}\left((N+1)Z_4 - 3(N-1)S_2^2\right) ,
\end{align*}
for $N>3$, where $S_2$ is defined by
\begin{align*}
S_2  :=  \frac{1}{6}\sum_{\substack{u\subset \{1,\ldots,4\}\\ |u|=2}}\frac{1}{N}\sum_{n=1}^{N}\prod_{j\in u}\left( f(x_n^{(j)})- \hat{\mu}^{(j)}\right).
\end{align*}
The unbiasedness of the estimators $\tilde{Z}_3$ and $\tilde{Z}_4$ can be proven in a way similar to that of Theorem~\ref{thm:unbias}, so that we omit it.
Moreover a bias correction can be established for any $m$ as long as $N\geq m$, although the corrected estimator shall be given in a complicated form when $m>4$.

%%%%%%%%%%%%%%%%%%%%%%%%%%%%%%%%%%%%%%%%%%%%%%%%%%%%%%%%%%%%%%%%%%%%%%%%%%%%%%%%%%%%%%%%%%%%%%%%%%%
%%%%%%%%%%%%%%%%%%%%%%%%%%%%%%%%%%%%%%%%%%%%%%%%%%%%%%%%%%%%%%%%%%%%%%%%%%%%%%%%%%%%%%%%%%%%%%%%%%%
%%%%%%%%%%%%%%%%%%%%%%%%%%%%%%%%%%%%%%%%%%%%%%%%%%%%%%%%%%%%%%%%%%%%%%%%%%%%%%%%%%%%%%%%%%%%%%%%%%%
\section{Numerical experiments}\label{sec:5}
Finally we conduct numerical experiments for simple test examples used in \cite{SH03}.
Here we aim to emphasize how easily one can implement our non-nested unbiased estimators without an additional pilot estimation step as done in \cite{SAS11}, and to demonstrate that our estimators can achieve the Monte Carlo rmse of order $C^{-1/2}$.
A performance comparison with other estimators including one by Sun et al. \cite{SAS11} is beyond the scope of this letter and shall be done in future work.

Three test examples due to Steckley and Henderson \cite{SH03} are given by:
\begin{description}
\item[Example~1] $Y\sim Beta(4,4)$, $X|(Y=y) \sim N(y,0.5)$, $f(x)=x$.
\item[Example~2] $Y\sim Beta(4,4)$, $X|(Y=y) \sim N(y,y^2)$, $f(x)=x$.
\item[Example~3] $Y \sim 1+Beta(4,4)$, $X|(Y=y) \sim \exp(1/y)$, $f(x)=x$.
\end{description}
We focus on the corrected estimator $\tilde{V}_3$ for estimating $\Var_Y\left[\EE_{X|Y}[f]\right]$ and its extensions $\tilde{Z}_3$ and $\tilde{Z}_4$ for estimating the third and fourth central moments of $\EE_{X|Y}[f]$, respectively.
All of the estimators $\tilde{V}_3,\tilde{Z}_3,\tilde{Z}_4$ can be implemented quite easily in MATLAB as follows.
Let us take Example~1 as an instance. $\tilde{V}_3$ can be written in only six lines as\\
\\
{\tt y = betarnd(4,4,N,1);\\
x1 = normrnd(y(:),sqrt(.5));\\
x2 = normrnd(y(:),sqrt(.5));\\
mu1 = mean(x1);\\
mu2 = mean(x2);\\
V = sum((x1(:)-m1).*(x2(:)-m2))/(N-1);
}\\
\\
It should be obvious that $\tilde{Z}_3$ and $\tilde{Z}_4$ are similarly written in several lines.

We run 1000 independent replications of our non-nested estimation for each setting, and compute the sample variance of those outputs.
Figure~\ref{fig:variance} shows the results of the variance of a conditional expectation for all the examples.
As mentioned in the first section, due to the non-nested forms of our estimators, we can expect that our estimators achieve the Monte Carlo rmse of order $C^{-1/2}$, which is in fact supported by these experimental results.
Figures~\ref{fig:third_moment} and \ref{fig:fourth_moment} show the results of the third and fourth central moments of a conditional expectation, respectively, for all the examples.
Again we can observe the Monte Carlo rmse of order $C^{-1/2}$.
As a future work, in order to enhance the performance of our estimators, it must be interesting to study how one can reduce the variance of our estimators, for instance, by using variance reduction techniques.
%%%%%%%%%%%%%%%%%%%%%%%%%%%%%%%%%%%%%%%%%%%%%%%%%%%%%%%%%%%
%%%%%%%%%%%%%%%%%%%%%%%%%%%%%%%%%%%%%%%%%%%%%%%%%%%%%%%%%%%

\begin{figure}
\begin{center}
\includegraphics[width=8cm]{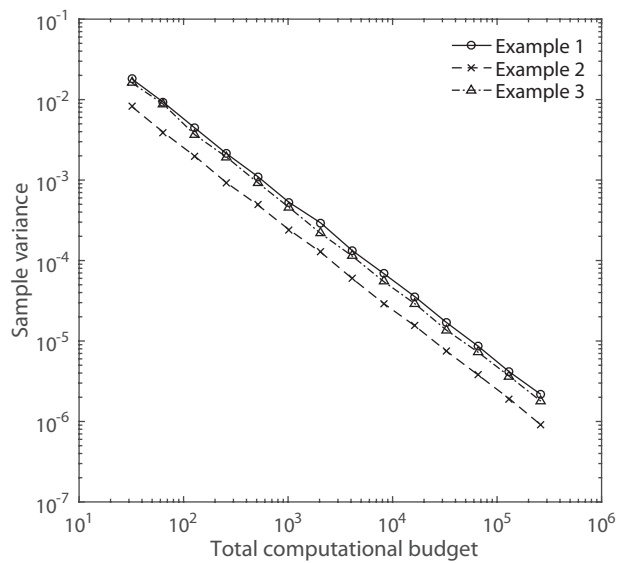}
\caption{Sample variance of $\tilde{V}_3$ for the variance of a conditional expectation with various total computational budgets.}
\label{fig:variance}
\end{center}
\end{figure}

\begin{figure}
\begin{center}
\includegraphics[width=8cm]{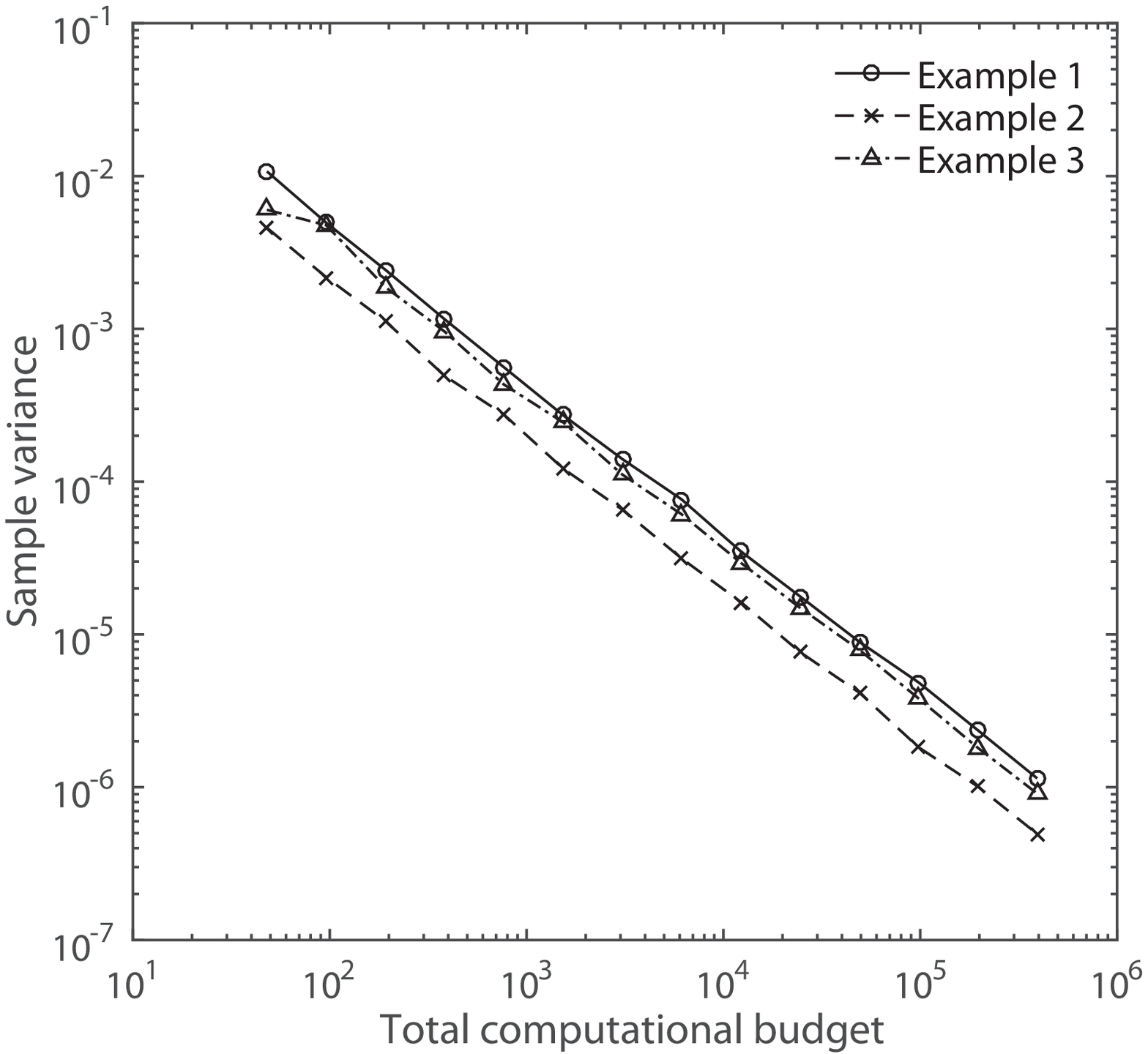}
\caption{Sample variance of $\tilde{Z}_3$ for the third central moment of a conditional expectation with various total computational budgets.}
\label{fig:third_moment}
\end{center}
\end{figure}

\begin{figure}
\begin{center}
\includegraphics[width=8cm]{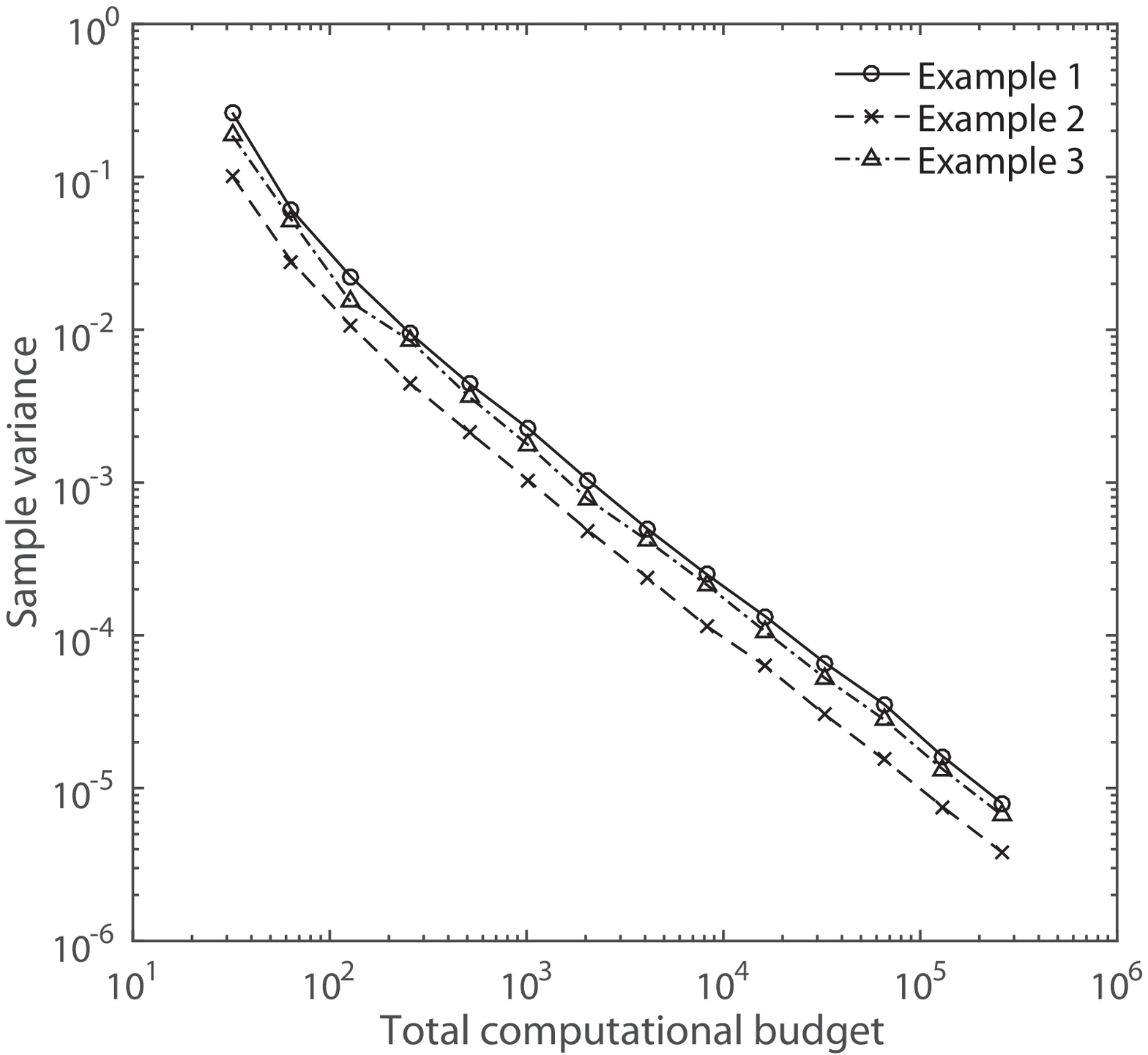}
\caption{Sample variance of $\tilde{Z}_4$ for the fourth central moment of a conditional expectation with various total computational budgets.}
\label{fig:fourth_moment}
\end{center}
\end{figure}


\begin{thebibliography}{99}
\bibitem{JKLNP14} Janon, A., Klein, T., Lagnoux, A., Nodet, M., Prieur, C.: Asymptotic normality and efficiency of two Sobol' index estimators, ESAIM: PS {\bf 18}, 342--364 (2014).
\bibitem{KFSM11} Kucherenko, S., Feil, B., Shah, N., Mauntz, W.: The identification of model effective dimensions using global sensitivity analysis, Reliab. Eng. Syst. Safety {\bf 96}, 440--449 (2011).
\bibitem{M02}   Mauntz, W.: Global sensitivity analysis of general nonlinear systems, Master's thesis, Imperial College, London (2002).
\bibitem{O13}   Owen, A.~B.: Variance components and generalized Sobol' indices, SIAM/ASA J. Uncertainty Quantification {\bf 1}, 19--41 (2013).
\bibitem{ODC14} Owen, A.~B., Dick, J., Chen, S.: Higher order Sobol' indices, Information and Inference {\bf 3}, 59--81 (2014).
\bibitem{Sal02} Saltelli, A.: Making best use of model evaluations to compute sensitivity indices, Comput. Phys. Comm. {\bf 145}, 280--297 (2002).
\bibitem{S90}   Sobol', I.~M.: Sensitivity estimates for nonlinear mathematical models, Matematicheskoe Modelirovanie {\bf 2}, 112--118 (1990) (in Russian), English translation in: Math. Model. Comput. Exp. {\bf 1}, 407--414 (1993).
\bibitem{S01}   Sobol', I.~M.: Global sensitivity indices for nonlinear mathematical models and their Monte Carlo estimates, Math. Comput. Simul. {\bf 55}, 271--280 (2001).
\bibitem{SH03} Steckley, S.~G., Henderson, S.~G.: A kernel approach to estimating the density of a conditional expectation, in: Proceeding of the 2003 Winter Simulation Conference, Chick, S., Sanchez, P.~J., Ferrin, D., Morrice, D.~J., eds., IEEE Press, Piscataway, NJ, 383--391 (2003).
\bibitem{SAS11} Sun, Y., Apley, D.~W., Staum, J.: Efficient nested simulation for estimating the variance of a conditional expectation, Operations Research {\bf 59}, 998--1007 (2011).
\bibitem{ZW03}  Zouaoui, F., Wilson, J.~R.: Accounting for parameter uncertainty in simulation input modeling, IIE Trans. {\bf 35}, 781--792 (2003).
\end{thebibliography}
\end{document}